%% file: Main.tex
\documentclass[conference]{IEEEtran}
\IEEEoverridecommandlockouts

\title{Fast Transition-Aware Reconfiguration of Liquid Crystal-based RISs}

\author{
\IEEEauthorblockN{Mohamadreza Delbari, Robin Neuder, Alejandro Jim\'{e}nez-S\'{a}ez, Arash Asadi, and Vahid Jamali}
\IEEEauthorblockA{
Technical University of Darmstadt (TUD), Darmstadt, Germany \vspace{-3mm}
\thanks{Delbari and Jamali’s work was supported in part by the Deutsche Forschungsgemeinschaft (DFG, German Research Foundation) within the Collaborative Research Center MAKI (SFB 1053, Project-ID 210487104) and in part by the LOEWE initiative (Hesse, Germany) within the emergenCITY center. Neuder and Alejandro Jim\'{e}nez-S\'{a}ez's work was supported by Deutsche Forschungsgemeinschaft (DFG, German Research Foundation) – Project-ID 287022738 – TRR 196 MARIE within project C09. Asadi's work was funded by the Deutsche Forschungsgemeinschaft within the mm-Cell project}
}
}

\input{1-Commands}

\begin{document}


\maketitle

\begin{abstract}

Liquid crystal (LC) technology offers a cost-effective, scalable, energy-efficient, and continuous phase tunable realization of extremely large reconfigurable intelligent surfaces (RISs). However, LC response time to achieve a desired differential phase is significantly higher compared to competing silicon-based technologies (RF switches, PIN diodes, etc). The slow response time can be the performance bottleneck for applications where frequent reconfiguration of the RIS (e.g., to serve different users) is needed. In this paper, we develop an RIS phase-shift design that is aware of the transition behavior and aims to minimize the time to switch among multiple RIS configurations each serving a mobile user in a time-division multiple-access (TDMA) protocol. Our simulation results confirm that the proposed algorithm significantly reduces the time required for the users to achieve a threshold signal quality. This leads to a considerable improvement in the achievable throughput for applications, where the length of the TDMA time intervals is comparable with the RIS reconfiguration time.
  
\end{abstract}

\section{Introduction}


Reconfigurable intelligent surfaces (RIS) have recently emerged in the wireless communication community as a rapidly advancing technology for realizing programmable radio environments \cite{di2019smart,yu2021smart}. These surfaces consist of sub-wavelength elements capable of dynamically altering the phases of reflected waves, thereby, e.g., establishing virtual links between base stations (BS) and mobile users \cite{Tang2021}. To have a strong and efficient virtual link, RIS must comprise a large number of elements \cite{najafi2021}. The generation of phase shifts can be achieved using various technologies, including silicon-based technologies (e.g., radio frequency (RF) switches and positive-intrinsic-negative (PIN) diodes) and liquid crystals (LC). In particular, fabrication in standard liquid crystal display (LCD) technology can be adopted for realizing large passive RISs, which brings the advantages of cost-effectiveness, scalability, low energy consumption, and continuous phase shifting. Due to these features, we focus on LC-RISs in this paper.

Several recent works have reported the use of LCs as a phase-shifting material for implementing RISs~\cite{ghannam2021reconfigurable, guirado2022mm}. Authors in~\cite{aboagye2022design,ndjiongue2021re} designed and optimized LCs in RISs for visible light communication systems. An experimental design for LC-RISs was reported in \cite{Robin2023Compact}. Jim{\'e}nez-S{\'a}ez et al. in \cite{jimenez2022liquid} also reviewed different characteristics of LCs such as cost, energy consumption, response time, and compared them with other phase-shifting technologies.


\begin{figure}[tbp]
    \centering
    \includegraphics[width=60mm]{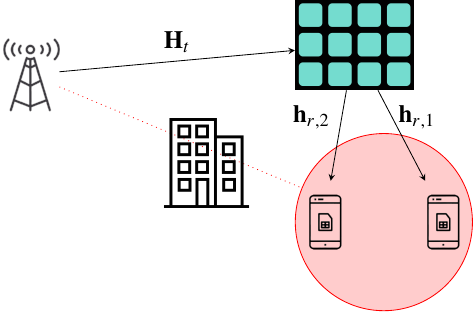}
        \caption{An RIS assists to establish a virtual link between a
transmitter and multiple receivers, while the direct channel is blocked by an obstacle.}
\vspace{-5mm}
    \label{fig:IRS}
\end{figure}

Despite various advantages, LC-RISs have a slower response time for achieving the desired differential phase compared to competing silicon-based technologies (RF switches, PIN diodes, etc). This slower transition time significantly affects the efficiency of beam switching capability, leading to the extended duration for phase-shift changes, which may result in performance falling below the limits accepted in cellular/wireless networks. In this paper, we study the effects of tuning time on the performance of LC-RISs and develop a novel RIS phase-shift design that accounts for the LC response time in order to enable fast transition among multiple phase-shift configurations each aiming to serve one mobile user. To the best of the authors' knowledge, this problem has not been studied in the literature before. 
Our main contributions are summarized as follows.

\begin{itemize}
    \item Firstly, we develop a physics-informed mathematical model for describing the LC response time based on the physical principles behind LC technology. 
    \item Next, we formulate an optimization problem for jointly designing multiple RIS phase-shift configurations which allow a fast transition from one configuration to the next. Due to the non-convexity of this problem, the global optimal solution is computationally prohibitive to attain. Therefore, we derive a suboptimal low-complexity solution, which is based on the Lagrange method.
\item Finally, we comprehensively evaluate the performance of the proposed RIS design. These results reveal that compared to transition-unaware benchmark schemes, the time required to reach a given signal quality threshold at mobile users can be significantly reduced by the proposed transition-aware phase-shift design.
\end{itemize}


\textit{Notation:} Bold capital and small letters are used to denote matrices and vectors, respectively.  $(\cdot)^\Trans$, $(\cdot)^\Herm$, $\circ$, and $\oslash$ denote the transpose, Hermitian, Hadamard product, and Hadamard division, respectively. Moreover, $\boldsymbol{0}_n$ and $\boldsymbol{1}_{n}$ denote a column vectors of size $n$ whose elements are all zeros and ones, respectively. $\bigO$ is the big-O notation.  $[\bX]_{m,n}$ and $[\bx]_{n}$ denote the element in the $m$th row and $n$th column of matrix $\bX$ and the $n$th entry of vector $\bx$, respectively.  
$\mathcal{CN}(\boldsymbol{\mu},\boldsymbol{\Sigma})$ denotes a complex Gaussian random vector with mean vector $\boldsymbol{\mu}$ and covariance matrix $\boldsymbol{\Sigma}$. 
$\Ex\{\cdot\}$ and $\Var\{\cdot\}$ represent expectation and variance, respectively. Finally $\Zset$, $\mathbb{R}$, and $\mathbb{C}$ represent the sets of integer, real and complex numbers, respectively. 

\section{System Model}
\label{Section 2}
In this paper, we consider a narrow-band downlink communication scenario comprising a BS with $N_t$ antenna elements, an RIS with $N$ LC-based unit cells, and $K$ single-antenna user equipment (UE). The users are served in a time-division multiple-access (TDMA) protocol. The received signal at user $k$ is given by:
\begin{equation}
\label{Eq:IRSbasic}
	y_k = \big(\bh_{d,k}^\Herm + \bh_{r,k}^\Herm \bGamma \bH_t \big) \bx_k +n_k,\quad k=1,\dots,K,
\end{equation}
where $\bx\in\Cset^{N_t}$ is the transmit signal vector, $y_k\in\Cset$ is the received signal vector at the $k$th UE, and  $n_k\in\Cset$ represents the additive white Gaussian noise (AWGN) at the $k$th UE, i.e., $n_k\sim\sCN(0,\sigma_n^2)$, where $\sigma_n^2$ is the noise power. The transmit vector $\bx_k\in\Cset^{N_t}$ can be written as $\bx_k=\bq_ks$, where $\bq_k\in\Cset^{N_t}$ is the beamforming vector for user $k$ and $s\in\Cset$ is the data symbol. Assuming $\Ex\{|s|^2\}=1$, the beamformer satisfies the transmit power constraint $\|\bq_k\|^2\leq P_t$.
Moreover,  $\bh_{d,k}\in\Cset^{N_t}, \bH_t\in\Cset^{N\times N_t}$, and $\bh_{r,k}\in\Cset^{N}$ denote the BS-UE, BS-RIS, and RIS-UE channel matrices, respectively. Furthermore, $\bGamma\in\Cset^{N\times N}$ is a diagonal matrix with main diagonal entries $[\bGamma]_n=[\bOmega]_n\e^{\jj[\bomega]_n}$ denoting the reflection coefficient applied by the $n$th RIS unit cell comprising phase shift $[\bomega]_n$  and reflection amplitude $[\bOmega]_n$.

\subsection{Channel model}
\label{Channel model}
We assume that an RIS is deployed such that it has a line of sight (LOS) link to both BS and users. Therefore, we model the links by Rician fading. For the ease of presentation, we present the channel model for a general multiple-input and multiple-output (MIMO) channel matrix $\bH\in\Cset^{N_\tx\times N_\rx}$ with $N_\tx$ and $N_\rx$ being the number of transmit and receive antennas, respectively, which can be then applied to $\bH_t$, $\bh_{r,k}$, and $\bh_{d,k},$ $\,\,\,\forall k$, with proper modification. A Rician MIMO channel model can be written as
$
\label{eq: channel model}
    \bH=\sqrt{\frac{K_f}{K_f+1}}\bH^{\mathrm{LOS}}+\sqrt{\frac{1}{K_f+1}}\bH^{\mathrm{nLOS}},
$
where $K_f$ denotes the $K$-factor
and determines the relative power of the LOS component to the non-LOS components of the channel. $\bH^\LOS$ and $\bH^\nLOS$ are given by
\begin{subequations}
    \begin{align}
	&\bH^{\mathrm{LOS}} = c\, \ba_{\rx}(\bPsi_\rx)\ba_{\tx}^\Herm(\bPsi_\tx), \label{Eq:LOS}\\
&[\bH^\nLOS]_{n,m}\sim\sCN(0,\sigma_\nLOS^2).\label{Eq:nLOS}
\end{align}
\end{subequations}

Here, $c^2$ and $\sigma_\nLOS^2$ determine the power of the LOS and nLOS channels, respectively, and $\ba_{\tx}(\cdot)\in\Cset^{N_\tx}$ and $\ba_{\rx}(\cdot)\in\Cset^{N_\rx}$ denote the Tx and Rx array steering vectors, respectively. Moreover, $\bPsi_\tx=(\theta_\tx,\phi_\tx)$ and $\bPsi_\rx=(\theta_\rx,\phi_\rx)$ are the Tx angle-of-departure (AoD) and Rx angle-of-arrival (AoA), respectively. Here $\theta_\tx$ and $\phi_\tx$ denote the elevation and azimuth angles, respectively. 

\subsection{LC phase shifter}
\label{Liquid crystal model}
The operational mechanism of LC relies on their electromagnetic anisotropy, where the electromagnetic properties of LC molecules differ based on their orientation relative to the RF electric field \cite{jimenez2022liquid}. Specifically, owing to the ellipsoidal shape of LC molecules, the LC exhibits higher permittivity (corresponding to greater phase shift) when the electric field $\vec{E}_{\rm RF}$ aligns with the major axis of the molecules $\vec{n}$, compared to alignment with the minor axis, as illustrated in Fig.~\ref{fig:V_phase}. Consequently, the manipulation of LC molecules' orientation provides a means to modify the phase shift of RF signals.
To achieve this effect, a slim layer of LC in the order of $\mu m$ is positioned between two electrodes. Under no voltage application, the molecules remain in their relaxed state, leading to $\vec{E}_{\rm RF}$ being perpendicular to $\vec{n}$ and resulting in the LC showing minimal relative permittivity $\varepsilon_{r,\perp}$. Conversely, applying the maximum voltage aligns the LC molecules with the induced external electric field, causing $\vec{E}_{\rm RF}$ to align parallel to $\vec{n}$ and thereby achieving a maximum relative permittivity $\varepsilon_{r,\parallel}$. The uppermost feasible phase shift, $\Delta\omega_{\max}$, scales with $\Delta n_\varepsilon\defeq\sqrt{\varepsilon_{r,\parallel}}-\sqrt{\varepsilon_{r,\perp}}$ as follows:
\vspace{-2mm}
\begin{equation}
    \Delta\omega_{\max}=2\pi l\Delta n_\varepsilon\frac{f}{c},
    \label{eq:omega epsilon}
\end{equation}
where $l$ is the phase-shifter length, $f$ is the frequency, and $c$ is the speed of light. In (\ref{eq:omega epsilon}), we used $\lambda=\frac{c}{f\sqrt{\varepsilon_r\mu_r}}$, where $\lambda$ denotes the wavelength and the relative permeability $\mu_r$ of LC is typically one \cite{garbovskiy2012liquid}.

\begin{figure}[tbp]
	\centering
    \includegraphics[width=65mm]{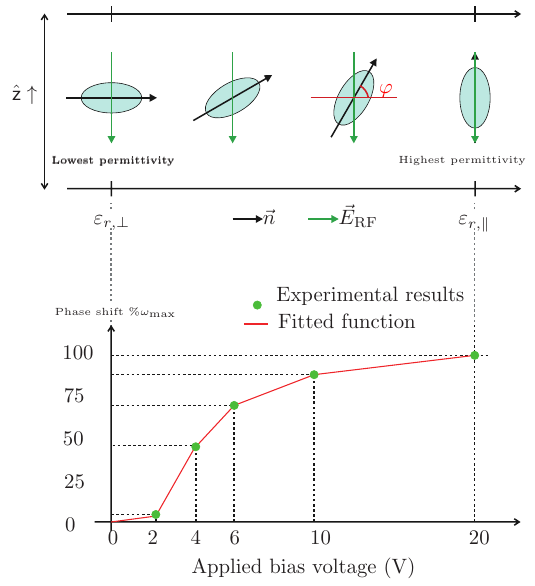}
    \caption{Phase shift vs. applied voltage for LC with 4.6 $\mu m$ LC layer thickness. The experimental data is taken from \cite{neuder2023architecture}, being fitted by a piecewise linear function.}
    \vspace{-5mm}
    \label{fig:V_phase}
\end{figure}
Fig.~\ref{fig:V_phase} also illustrates the relationship between the applied voltage and the phase shift. The phase shift can be smoothly adjusted from $0$ to $\omega_{\mathrm{max}}$.
The applied bias voltage is alternating current (AC) with frequency 1~kHz because applying direct voltage would degrade the LC performance \cite{xu2014image}. For future reference, we define the function $\bomega=f(\bv)\label{eq:fv}$ describing the relation between changing phase and voltage,
where $\bv\in\Rset^N$ and $\bomega\in\Rset^N$ are applied voltages and corresponding phase shifts of all RIS elements.



\section{Transition-Aware LC-RIS Phase-shift Design}
\label{Transition-Aware LC-RIS}

In this section, we first develop a model for characterizing LC-RIS response time. Subsequently, we formulate an optimization problem to jointly design multiple phase-shift configurations that enable fast transition among them.

\subsection{LC-RIS response time modeling}
\label{Time response of LC}
The response time of LCs is different for positive and negative changes of phase shift \cite{nobles2021eight}. In particular, for positive phase shifts, the electric field caused by the applied voltage determines the response time. On the contrary, for the negative phase shifts, the anchoring forces in LC in the absence of the voltage are the determining factors for response time.

 \textbf{Mathematical characterization:} 
 As discussed, the phase shift $\omega$ applied by an LC unit cell is directly related to the orientation angle of the LC molecules, denoted by $\varphi$ which can be controlled by an applied electric field $\vec{E}_{\rm RF}$ (i.e., an applied voltage). The value of $\varphi$ is a function of time, $t$, and the LC molecule distance to the electrodes, $z$ (see Fig.~\ref{fig:V_phase}), and is characterized for an AC control voltage by the Erickson-Leslie differential equation \cite{ericksen1961conservation,leslie1968some}:
\vspace{-1mm}
\begin{IEEEeqnarray}{ll}(K_{11}\cos^2\varphi+K_{33}\sin^2\varphi)\frac{\partial^2 \varphi}{\partial z^2}+(K_{33}-K_{11})\sin\varphi \nonumber\\ \qquad \times \cos\varphi(\frac{\partial\varphi}{\partial z})^2 
    +\varepsilon_0\Delta\varepsilon E^2\sin\varphi\cos\varphi=\gamma_1(\frac{\partial\varphi}{\partial t}).
\end{IEEEeqnarray}
Here, $\gamma_1$ is the LC's rotational viscosity, $K_{11}$ and $K_{33}$ represent the LC splay and bend elastic constants, respectively, $\varepsilon_0\Delta\varepsilon E^2$ is the electric field energy density, $\varepsilon_0$ is vacuum permittivity, and $\Delta\varepsilon$ is the LC dielectric anisotropy. In general, this equation cannot be solved analytically, however, using the $K_{11}\approx K_{33}$ approximation, Jakeman and Raynes calculated the reorientation times of the LC director when the electric field is switched off, i.e., $E=0$, in \cite{jakeman1972electro}. In this case, we obtain:
\begin{equation}
\label{eq:exponential function g}
    \varphi(z,t)\approx \varphi_0\sin(\frac{\pi z}{l_g})\exp{(-\frac{t}{\tau_0})},
\end{equation}
where $\varphi_0$ is the maximum tilt angle of the LC directors (i.e., at $z=\frac{l_g}{2}$) at the initial time $t=0$, $l_g$
is the LC thickness, and $\tau_0$ is the LC
director reorientation time constant, which is a function of LC parameters $\gamma$, $K_{11}$, and $K_{33}$. When $E$ is nonzero, with further assumption, (i.e., $\sin\varphi\cos\varphi\approx\varphi(1-\frac{\varphi^2}{2})$), authors in \cite{jakeman1972electro} derived a similar exponential solution which has a different time exponent and is omitted here due to space constraints, see \cite{wang2005studies}. For future reference, we denote the time exponents in the absence and presence of the bias voltage by $\tau_0^-$ and $\tau_0^+$, respectively, where $\tau_0^-\gg \tau_0^+$ holds \cite{neuder2023architecture}, see Fig.~\ref{fig:time response+-}.

For RISs, not the absolute value of the phase but rather the changes in the phase is important. In order to mathematically model the transition, we define a function $g$ for showing how the RIS phase-shift $\omega$ is temporally evolved when the initial phase shift is $\omega_0$ and the desired phase is $\omega_d$
\begin{equation} \label{omega transition}
    \omega(t)=g(t;\omega_0,\omega_d).
\end{equation}
Note that function $g$ should satisfy the following asymptotic property:
\begin{equation}
\label{eq: exponential function}
    \omega_d= \lim_{t \to +\infty} g(t;\omega_0,\omega_d)=f(v_d).
\end{equation}
where $f(\cdot)$ is given in \eqref{eq:fv}, see Fig.~\ref{fig:V_phase}. 
Based on the solution in (\ref{eq:exponential function g}) and the fact that the phase-shift $\omega$ of one LC molecule keeps the same form of the tilt angle $\varphi$ of it \cite{wang2005studies,wu1990experimental} but with different time constants in two different cases rise and decay,  function $g$ is an exponential function that assumes different time exponents for positive and negative changes in phase-shift and must satisfy (\ref{eq: exponential function}). This leads to the following response time for LC unit-cell $n, \forall n$:
\begin{equation}
\label{eq:general time response}
    [\bomega(t)]_n=g(t;[\bomega_0]_n,[\bomega_d]_n)=[\bomega_d]_n+([\bomega_0]_n-[\bomega_d]_n)e^{\frac{-t}{[\btau_{c}]_n}},
\end{equation}
where $\btau_{c}\in\{\tau_c^+,\tau_c^-\}^N$ and its $n$th element is equal to $\tau_c^+$ if $[\bomega_d]_n\geq [\bomega_0]_n$ and is equal to $\tau_c^-$ if $[\bomega_d]_n < [\bomega_0]_n$ \cite{wang2005studies}.
\begin{figure}
    \centering
\includegraphics[width=65mm]{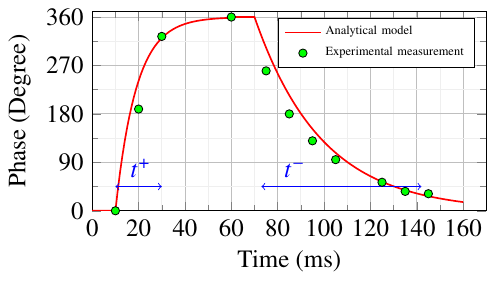}
    \caption{Experimental result of an LC phase shifter response time (green points, $t^+\approx15$~ms, $t^-\approx72$~ms) \cite[Fig.~4a]{neuder2023architecture}.}
    \vspace{-5mm}
    \label{fig:time response+-}
\end{figure}

\textbf{Over-and undershooting technique:} In order to achieve a faster transition, we apply the so-called over- and under-shooting technique \cite{hassanein2019optimizing} described in the following. For increasing phase-shift, i.e., $[\bomega_d]_n>[\bomega_0]_n$, instead of directly applying the desired voltage control $[\bv_d]_n=f^{-1}([\bomega_d]_n)$, we apply the maximum voltage $[\bv]_n=f^{-1}(\omega_\tmax)$ and when the desired phase shift is achieved, change the voltage control to the desired value $[\bv_d]_n$.
The time required for $n$th cell to reach the desired phase shift can be computed by replacing in Eq. (\ref{eq:general time response}) and solving the following equation for $t$:
\begin{equation}
    [\bomega_d]_n=\omega_\tmax+([\bomega_0]_n-\omega_\tmax)e^{\frac{-t}{\tau^+_c}},
\end{equation}
which leads to
\begin{equation}
    \label{eq: time positive}    [\bt_r]_n=\tau^+_c\ln{(\frac{\omega_{\text{max}}-[\bomega_0]_n}{\omega_{\text{max}}-[\bomega_d]_n})}=\tau^+_c\ln{(1+\frac{[\bomega_d]_n-[\bomega_0]_n}{\omega_{\text{max}}-[\bomega_d]_n})},
\end{equation}
where $[\bt_r]_n$ is the time where phase of $n$th cell reaches its desired phase. 
Similarly, when the phase shift should be decreased, we can also use the undershooting technique by first setting the voltage to zero and applying the desired voltage when the phase shift is reached. This results in the following tuning time when $[\bomega_d]_n<[\bomega_0]_n$ after replacing in Eq. (\ref{eq:general time response}):
\begin{equation}
    \label{eq: time negative}
    [\bt_r]_n=\tau^-_c\ln{(\frac{[\bomega_0]_n}{[\bomega_d]_n})}=\tau^-_c\ln{(1+\frac{[\bomega_0]_n-[\bomega_d]_n}{[\bomega_d]_n})}.
\end{equation}

Eqs.~\eqref{eq: time positive} and \eqref{eq: time negative} describe the transition behavior of LC unit cells to achieve a desired phase shift and are used in our simulation results for performance assessment.


\subsection{Transition-aware phase-shift design}
\label{problem formulation}

Next, assuming a TDMA protocol, we design multiple LC-RIS phase-shift configurations that allow a fast transition from one configuration to the next each serving one user. As a quality of service (QoS), we consider the signal-to-noise ratio (SNR) for user $k,\,\,\forall k$, defined as
\vspace{-1mm}
\begin{IEEEeqnarray}{ll}
\label{eq: SNR}
    \SNR_k=\frac{|(\bh_k^\eff)^\Herm\bq_k|^2}{\sigma^2_n},\vspace{-2mm}
\end{IEEEeqnarray}
where $(\bh_k^\eff)^\Herm=\bh_{d,k}^\Herm + \bh_{r,k}^\Herm \bGamma_k \bH_t$ is the end-to-end channel for the $k$th user accounting for the impact of the RIS. For simplicity, we assume the amplitudes of the reflection coefficients are one, i.e., $|[\bGamma_k]_n|=[\bOmega]_n=1,\,\,\forall n$.

By convention, we assume users are served in order of their indices in one-time frame, i.e., first user 1, then user 2, until user $K$, and repeat this procedure, see Fig. \ref{fig:user_TDMA}. The duration of the time frame depends on the delay requirement. The update rate of the LC-RIS phase-shift (i.e., the number of time frames the phase-shifts remain valid) depends on the users' velocity and direction. Eqs. (\ref{eq: time positive}) and (\ref{eq: time negative}) show that the higher the change in phase shifts, the higher the required tuning time. In this paper, we aim at achieving a given required SNR threshold per user, denoted by $\gamma_k^\thr$, by applying the minimum changes in the phase shift from previous phase-shift reconfiguration. To mathematically formulate the problem, we define for $k=1,..., K$:
\vspace{-3mm}
\begin{IEEEeqnarray}{ll}
\Delta\bomega_k\defeq
\begin{cases}
    \bomega_k-\bomega_{k-1},\quad &\mathrm{if}\,\, k\neq 1\\
    \bomega_{1}-\bomega_K,\,\,&\mathrm{if}\,\, k=1.
\end{cases}\vspace{-1mm}
\end{IEEEeqnarray}

Since directly minimizing the required tuning time in (\ref{eq: time positive}) and (\ref{eq: time negative}) is challenging, we aim at minimizing the phase-shift changes when reaching a QoS while accounting for the fact that increasing and decreasing phase-shifts occur with different time constants. In particular, we adopt as an optimization criterion the sum of weighted differential phase shifts $\sum_{k=1}^K\|\bc_{k}\circ\Delta\bomega_k\|^2$, where $\bc_{k}\in\{c^+,c^-\}^N$ is factor obtained as:\vspace{-2mm}
\begin{align}
    [\bc_{k}]_n = \begin{cases}
    c^+,\quad &\mathrm{if}\,\,[\Delta\bomega_k]_n\geq 0\\
    c^-,\quad &\mathrm{if}\,\,[\Delta\bomega_k]_n < 0.
    \end{cases}\vspace{-2mm}
\end{align}
Since decreasing LC phase shifts is more time-consuming than increasing them (i.e., $\tau_c^-\gg\tau_c^+$), it is reasonable to assign weights such that $c^-\gg c^+$ holds. 
Using this cost function, we formulate the following optimization problem:
\vspace{-2mm}
\begin{subequations}
\label{eq:optimization 1}
\begin{align}
    \text {P1:}\quad&~\underset{\bomega_k,\bq_k,\,\,\forall k}{\min}~\sum_{k=1}^K\|\bc_{k}\circ\Delta \bomega_k\|_2^2
    \\&~\text {s.t.} ~~\SNR_k\geq \gamma_k^\thr,\,\, \forall k
    \\&\quad\hphantom {\text {s.t.} } 0\leq [\bomega_k]_n < \omega_\tmax, \forall n, k,
    \\&\quad\hphantom {\text {s.t.} } \|\bq_k\|_2^2\leq P_t, \forall k.
\end{align}
\end{subequations}

Here, (\ref{eq:optimization 1}b) is the user SNR constraint, (\ref{eq:optimization 1}c) is the realizable phase-shift range with $\omega_\mathrm{max}=2\pi$ being the maximum unit-cell phase-shift, (\ref{eq:optimization 1}d) forces the transmit power constraint. Problem P1 is a non-convex problem optimization because of the constraint (\ref{eq:optimization 1}b). $\bq_k$ indirectly influences the cost function via the SNR in (\ref{eq:optimization 1}b). Due to the non-convexity of the problem and the coupling of variables $\bq_k$ and $\bomega_k$ in (\ref{eq: SNR}), obtaining a globally optimal solution is computationally challenging. In the following, we obtain an efficient sub-optimal solution using alternative optimization (AO). 

\begin{figure}[t]
	\centering
    \includegraphics[width=0.49\textwidth]{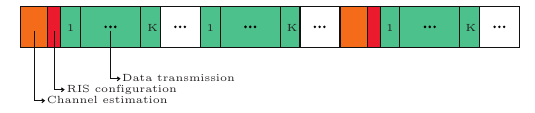}
    \caption{Block diagram of the proposed RIS-assisted downlink transmission
scheme, including the sub-blocks of channel estimation, and RIS phase-shift design.}
    \label{fig:user_TDMA}
    \vspace{-5mm}
\end{figure}

\textbf{Beamformer design:} When $\bomega_k$ is given and fixed for each user, problem $\text{P1}$ reduces to a feasibility problem because $\bq_k$ does not appear in the cost function. Alternatively, one can solve the following equivalent sub-problem:
\vspace{-2mm}
\begin{subequations}
\label{eq:optimization 2}
\begin{align}
    \text {P2:}\quad&~ \underset{\bq_k}{\max}~\SNR_k
    \\&~\text {s.t.} ~~\|\bq_k\|_2^2\leq P_t.
\end{align}
\end{subequations}
This is a standard SNR-maximization problem in multi-input single-output (MISO) systems, which  has the following match-filter precoder solution \cite{tse2005fundamentals}: 
\begin{equation}
\label{eq: SNR with LOS}
 \bq_k = \frac{\sqrt{P_t}}{\|\bh_k^\eff\|} \bh_k^\eff 
\overset{(a)}{\approx}  \frac{\sqrt{P_t}}{\|\ba_\BS(\bPsi_\BS)\|} \ba_\BS(\bPsi_\BS),
\end{equation}
where the problem is feasible if $P_t\frac{\|\bh_k^\eff\|^2}{\sigma_n^2}\geq \gamma_k^\thr,\,\,\forall k$. Approximation $(a)$ holds for typical RIS deployment for coverage extension, when the direct channel BS-UE is blocked, i.e., $\bh_{d,k}\approx\boldsymbol{0}_{N_t}$, and the BS-RIS link is LOS dominant with $\bPsi_\BS$ and $\bPsi_\RIS$ denoting the AoD from BS and the AoA on the RIS, respectively  \cite{Zhang2020Capacity}. This LOS solution, denoted by $\bq_\LOS$, does not depend on $\bomega_k$, and hence does not change in the AO algorithm. For the remainder of this paper, we focus on  the RIS deployment scenario for coverage extension and adopt beamformer $\bq_\LOS$.

\textbf{RIS design:} Now, we assume the beamformer $\bq_\LOS$ is fixed and optimize $\bomega_k$. The phase configuration problem is:
\vspace{-2mm}
\begin{subequations}
\label{eq:optimization 3}
\begin{align}
    \text {P3:}\quad&~\underset{\bomega_k,\,\,\forall k}{\min}~\sum_{k=1}^K\|\bc_{k}\circ\Delta \bomega_k\|_2^2
    \\&~\text {s.t.} ~~\SNR_k\geq \gamma_k^\thr,\,\, \forall k
    \\&\quad\hphantom {\text {s.t.} } 0\leq [\bomega_k]_n < 2\pi, \forall n, k,
\end{align}
\end{subequations}
where the SNR can be rewritten as
\begin{equation}\label{eq:SNRquad}
    \SNR_k=\frac{\|\bbm_k^\Herm\bs_k\|^2}{\sigma_n^2}\|\bq_\LOS\|^2\|\ba_\BS^\Herm(\bPsi_\BS)\|c^2=\bs_k^\Herm\bbM_k\bs_k,
\end{equation}
with $\bs_k=[\e^{j[\bomega_k]_1}, ..., \e^{j[\bomega_k]_N}]^\Trans$ and $\bbM_k=c^2\frac{\|\ba_\BS^\Herm(\bPsi_\BS)\|\|\bq_\LOS\|^2}{\sigma_n^2}$ $\bbm_k\bbm_k^\Herm$ 
with $\bbm_k=\mathrm{diag}(\bh_{r,k}^\Herm)\ba_{\RIS}(\bPsi_\RIS)$ where $c^2$ is defined as the power gain of the BS-RIS link.

Problem P3 is non-convex due to the non-convex SNR function in \eqref{eq:SNRquad} in variable $\bomega_k,\,\,\forall k$, and the dependency of $\bc_{k}$ on $\Delta \bomega_k$ in the objective function. In the following, we develop a sub-optimal solution to P3  using the Lagrange formulation. First let us define:
\vspace{-2mm}
\begin{align}
\label{eq:Lagrange definition}
&L_{\text{total}}(\bW,\blambda)\defeq\sum_{k=1}^K\|\bc_{k}\circ\Delta \bomega_k\|_2^2+\sum_{k=1}^K\lambda_k(\gamma_k^\thr-\SNR_k)\nonumber\\
&=\sum_{k=1}^K\big(\|\bc_{k}\circ\Delta \bomega_k\|_2^2+\lambda_k(\gamma_k^\thr-\SNR_k)\big),
\end{align}
where $\bW=[\bomega_1, \cdots, \bomega_K]$ and $\blambda=[\lambda_1, \cdots, \lambda_K]^\Trans$ with $\lambda_k\geq0, \,\forall k$ is the Lagrange multiplier associated with the $k$th user inequality constraint. Our approach is decoupling Lagrangian on each user. Defining $L(\bomega_k,\lambda_k)\defeq\|\bc_{k}\circ\Delta \bomega_k\|_2^2+\lambda_k(\gamma_k^\thr-\SNR_k)$, we obtain the dual problem as
\vspace{-2mm}
\begin{equation}
\label{eq:Lagrange problem}
    \text {P4:}\quad~\underset{\blambda}{\max}\underset{\bW\in\Wset^{N\times K}}{\min}~\sum_{k=1}^KL(\bomega_k,\lambda_k),
\end{equation}
where $\Wset=\{\omega\in\Rset\mid0\leq\omega<2\pi\}$. Until now, we decoupled the problem on each user and wrote the duality Lagrangian of the problem after this decoupling. The dual function yields lower bounds on the optimal value of the problem. Here, all phase-shift variables are coupled to each other in $\SNR_k,\,\,\forall k$. To find an analytical solution to this problem, we try to decouple them on each RIS element, too, in the following two steps. At first, we derive the deviation of $L(\bomega_k,\lambda_k)$, then apply an approximation, and finally, we take the integral of~it.


\textit{1- Decoupling on each element $n$:} First we calculate gradient of $L(\bomega_k,\lambda_k)$ for $k$th user:
\begin{equation}
\begin{split}
&\nabla_{\bomega_k} L(\bomega_k,\lambda_k) = 2\bc_{k}\circ\Delta\bomega_k + \left(\frac{\partial L(\bomega_k,\lambda_k)}{\partial \bs}\right)\circ(\nabla_{\bomega_k} \bs) \\
&\quad\quad\quad\quad\quad\quad + \left(\frac{\partial L(\bomega_k,\lambda_k)}{\partial \bs^\Herm}\right)\circ(\nabla_{\bomega_k} \bs^\Herm) \\
& =2\bc_{k}\circ\Delta\bomega_k - \lambda_k\left((\bz_k^\Herm)^\Trans\circ(\jj\bs_k) + (\bz_k)\circ(-\jj\bs_k^\Herm)^\Trans\right),
\end{split}
\label{eq:lagrangian derivative}
\end{equation}
where $\bz_k^\Herm\defeq\bs_k^\Herm\bbM_k$ and note $\bbM^\Herm=\bbM$. Here, $\bs_k$ contains phase shifts for $k$th user is our variable but $\bz_k$ is the weighted average of all $[\bs_k]_n,\,\,\forall n$. Our observation for the channels with LOS dominant link is that applying a small change in each phase shift $[\bomega_k]_n$ does not lead to a sensible change in $\bz_k$. Exploiting this property, the following lemma can be derived:
\begin{lem}\label{lem: decoupling on elements}
    Assuming a small change in each $[\bomega_k]_n$ does not lead to a change in $\bz_k$. Under this condition the following decomposition is correct:
\begin{equation} \label{eq:after gradient theorem}
    L(\bomega_k,\lambda_k)\!=\!\!
    \int \nabla_{\bomega_k} L(\bomega_k,\lambda_k)  \bd\bomega_k \approx \sum_{n=1}^N L_n([\bomega_k]_n,\lambda_k),
\end{equation}
where $L_n([\bomega_k]_n,\lambda_k)=[\bc_k]_n[\Delta\bomega_k]^2_n-2\lambda_k[\br_k]_n\cos([\bomega_k]_n-[\bphi_k]_n),\,\,\forall n$, and each element of $[\br_k]_n$ and $[\bphi_k]_n$ are the absolute and angle of $[\bz_k]_n$, respectively.
\end{lem}

\begin{proof}
The following is valid $\forall n$:
    \begin{align}
    &[\bz_k^\Herm]_n(\jj[\bs_k]_n)+[\bz_k]_n(-\jj[\bs_k^\Herm]_n) \nonumber\\
    &=\jj[\br_k]_n\e^{-\jj[\bphi_k]_n}\e^{\jj[\bomega_k]_n}-\jj[\br_k]_n\e^{\jj[\bphi_k]_n}\e^{-\jj[\bomega_k]_n} \nonumber\\
    &=2[\br_k]_n\sin{([\bphi_k]_n-[\bomega_k]_n)}.
    \end{align} 
    Now we can take the integral for each element:
\begin{align}
\label{eq: L_n}
\int &[\nabla_{\bomega_k} L(\bomega_k,\lambda_k)]_n\, d[\bomega_k]_n = ([\bc_k]_n([\bomega_k]_n-[\bomega_{k+1}]_n))^2 \nonumber \\
&\quad -2\lambda_k[\br_k]_n\cos([\bomega_k]_n-[\bphi_k]_n), \,\, \forall n,
\end{align}
where we omitted the constant in output because it will not affect the optimization. This concludes the proof.
\end{proof}



With the help of Lemma.~\ref{lem: decoupling on elements}, we decoupled problem P4 on each element. Now equivalently, we can find the minimizer of each $L_n([\bomega_k]_n,\lambda_k)$ to minimize the summation of them. We do it iteratively in order to satisfy our assumption in lemma.~\ref{lem: decoupling on elements}. 

\textit{2- Iteratively finding the minimizer of each $L_n([\bomega_k]_n,\lambda_k)$:} In each iteration (i.e. $i$th iteration), our assumption in Lemma~\ref{lem: decoupling on elements} is valid if each element of $\bphi_k^{(i)}$ changes in a range of $[-\delta_k,\delta_k]$ compared to its previous phase shift $\bphi_k^{(i-1)}$ with $\delta_k$ being a small real constant. We name the minimizer of $L_n([\bomega_k]_n,\lambda_k)$ as $[\bomega_k]_n^g$. This point can be found by line search in this range or in more efficient way we can reduce it to comparing some special points. All the stationary points of $L_n([\bomega_k]_n,\lambda_k)$ as $[\bomega_k]_{n,m}^*,\,\,\forall m\in\Zset$. $[\bomega_k]_n^g$ is one of the points in set $\{0,2\pi,-\delta_k,\delta_k,[\bomega_k]_{n,m}^*\}\cap[-\delta_k,\delta_k]$\footnote{The proposed iterative approach is not exactly the same as the gradient descend since instead of moving along the gradient, in each iteration, we find the minimizer of $L_n([\bomega_k]_n,\lambda_k)$ in a small interval of length $2\delta_k$.} 



\textbf{Algorithm and complexity analysis:} The proposed algorithm to obtain a suboptimal solution to P1 is summarized in Alg. \ref{alg:cap}. In each iteration, we focus on one user and change its phase shift in a small range to reduce the transition time from $\bomega_{k-1}$ to $\bomega_{k}$. 
After finding a new phase-shift for $k$th user, we check the SNR constraint. If the constraint is still valid, with decreasing $\lambda_k$ we try to find other potential solutions, however, if the constraint is not valid anymore, increasing $\lambda_k$ prioritises the SNR constraint.
In this algorithm, the dominant complexity is related to finding $\bomega_k^g$ in each iteration. Assume the complexity of line search over the range is $L$ then the dominant complexity is $\bigO(I_\tmax NKL)$.

\begin{algorithm}[t]
\caption{Proposed Algorithm for Problem (P4)}\label{alg:cap}
\begin{algorithmic}[1]
\STATE \textbf{Initialize:} Estimate $\bbM_k,\,\bomega_k^{(0)}= \bphi^\LOS_k,\, [\bs_k]_n^{(0)}=\e^{\jj[\bomega_k^{(0)}]_n},\,\forall n,\,\br^{(0)}_k\e^{-\jj \bphi^{(0)}_k}=\bbM_k\bs_k^\Herm$, $\,\lambda^{(0)}_k>0$, $\bq_k=\bq_\LOS$, $0<\delta_k<\pi$, $\forall k=1, \cdots, K$. Set $0<\alpha<1$, and $I_{\tmax}$.
\FOR{$i=1:I_{\tmax}$}
    \FOR{$k=1:K$}

    \STATE Find $[\bomega_k]^{(i)}_n=[\bomega_k]_n^g\in(-\delta_k,\delta_k)$ as a minimizer of $L_n([\bomega_k]_n,\lambda_k),\,\forall n$ in Eq.~\ref{eq: L_n} with line search.
    \IF{$\SNR^{(i)}_k<\gamma_k^\thr$}
    \STATE Update $\bomega^{(i)}_k=\bomega^{(i-1)}_k$, and $\lambda_k^{(i)}=\frac{\lambda_k^{(i-1)}}{\alpha}$.
    \ELSE
    \STATE Update $\lambda_k^{(i)}=\alpha\lambda_k^{(i-1)}$.
    \ENDIF
    \STATE Update $[\bs_k]_n^{(i)}=\e^{\jj[\bomega_k^{(i)}]_n},\,\forall n$, $\br^{(i)}_k\e^{-\jj \bphi^{(i)}_k}=\bbM_k\bs_k^\Herm$.
    \ENDFOR
\ENDFOR
\end{algorithmic}
\end{algorithm}







\section{Performance Comparison}
\label{simulation result}

\subsection{Simulation Setup}
We employ the simulation configuration for coverage extension illustrated in Fig. \ref{fig:IRS}. We assume there are two users located in different directions, $(\theta,\phi)=(-10^\circ,33^\circ)$ and $(\theta,\phi)=(-10^\circ,-33^\circ)$. The BS features a $4\times4=16$ UPA positioned at $[30,0,10]$~m. The RIS comprises a UPA located at $[0,50,5]$~m, consisting of $N_\y\times N_\z$ square tiles along the $\y$- and $\z$-axes, respectively. The element spacing for UPAs at both the BS and RIS corresponds to half the wavelength. The UEs have a single antenna. The noise variance is computed as $\sigma_n^2=WN_0N_{\rm f}$ with $N_0=-174$~dBm/Hz, $W=20$~MHz, and $N_{\rm f}=6$~dB. We assume $28$~GHz carrier frequency, $\beta=-61$~dB at $d_0=1$~m, and $\gamma_\thr=10$~dB. Moreover, we adopt $\eta = (3.5,2,2)$ and  $K_f=(0,10,10)$ for the BS-UE, BS-RIS, and RIS-UE channels, respectively. The main focus of this analysis is the SNR comparison. As a benchmark, we consider RIS design based on anomalous reflection \cite{najafi2021} which is unaware of the transition behavior of the LC-RIS. Also for the factor of each element in response time, we choose $c^-\propto\sqrt{\tau_c^-}$ and $c^+\propto\sqrt{\tau_c^+}$ for negative and positive changes, respectively such that the cost function in Eq. (\ref{eq:optimization 1}a) linearly scales with time constants like in Eqs. (\ref{eq: time positive}) and (\ref{eq: time negative}). The parameters used in this simulation are $P_t=47$~dBm, $\alpha=0.985$, $I_\tmax=100$, and $\delta_k=\frac{\pi}{8},\,\forall k$.

\subsection{Simulation Result}
Fig.~\ref{fig:TDMA_time} shows the SNR (dB) when the RIS configuration is switched every 60 ms to serve the users. This specific time interval was chosen to demonstrate the performance of the proposed algorithm in real-time conditions. We observe in the figure, that our proposed algorithm achieves the SNR requirement (i.e., $10$ dB) much faster than the benchmark which simply maximizes the final SNR without accounting for the transition time. This confirms that the proposed algorithm finds the phase-shift configuration that meets the SNR requirement without significantly changing the existing RIS phase-shifts. In contrast, the benchmark scheme simply finds the phase-shift configuration that maximizes the SNR which often implies a significant change of the RIS phase-shifts compared to the current phase-shifts. In Fig.~\ref{fig:Percentage_rate}, we plot the effective data rate is given by
\vspace{-2mm}
\begin{equation}
  \label{eq:rate}
    R=\frac{\max(T_s-T_c,0)}{T_s}\log{(1+\mathrm{SNR_{thr}})},
\end{equation}
where $T_c$ is the time RIS needs to reconfigure and reach $\mathrm{SNR_{thr}}$ and $T_s$ is the time interval that RIS switches between serving the users (determined by the application scenario, e.g., delay requirement, users' mobility). When $T_s$ is a very small number, both algorithms cannot serve users because the required time to change their phase configuration is not enough. With increasing $T_s$, by applying the proposed algorithm, the user's SNR surpasses the threshold faster. This increases the ratio $\frac{\max(T_s-T_c,0)}{T_s}$, leading to higher data rates. With $T_s$ approaching infinity, this proportion goes to $\log{(1+\mathrm{SNR_{thr}})}$ for all approaches because there is enough time to reach the final phase configuration.



\begin{figure}
    \centering
    \includegraphics[width=0.5\textwidth]{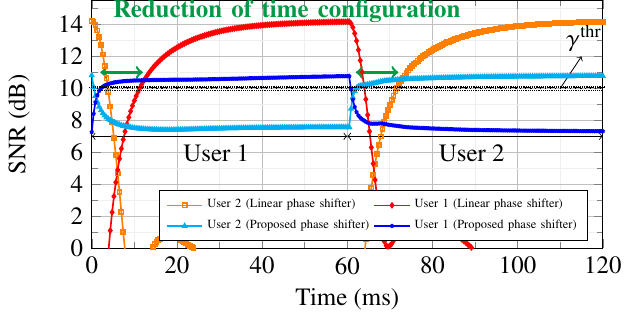}\vspace{-3mm}
    \caption{SNR (dB) comparison between linear phased shifting and proposed algorithm with 60 ms serving the users.}
    \label{fig:TDMA_time}
    \vspace{-3mm}
\end{figure}

\begin{figure}
    \centering
    \includegraphics[width=0.5\textwidth]{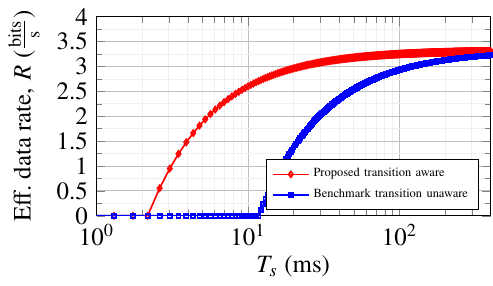}\vspace{-3mm}
    \caption{The effective data rate ($R$) versus switching time between two users, $T_s$, for two different algorithms.}
    \label{fig:Percentage_rate}
    \vspace{-3mm}
\end{figure}

\section{Conclusion}
\label{conclusion}
In this study, we have first modeled the response time of LC-RISs. Subsequently, we have developed an algorithm to jointly design the fast RIS reconfiguration for serving $K$ users in a TDMA manner. The simulation results demonstrated the superiority of the proposed transition-aware design over a transition-aware benchmark from the literature.



\input{2-References}

\end{document}

%% file: 1-Commands.tex
\usepackage{amsfonts,amsmath,amsthm,bm}
\usepackage{cite}
\usepackage{amsmath,amssymb,amsfonts}
\usepackage{algorithmic}
\usepackage{graphicx}
\usepackage{textcomp}
\usepackage{xcolor}
\def\BibTeX{{\rm B\kern-.05em{\sc i\kern-.025em b}\kern-.08em
    T\kern-.1667em\lower.7ex\hbox{E}\kern-.125emX}}
 \usepackage[T1]{fontenc}

\usepackage{eucal}
\usepackage{algorithm}
\usepackage{algorithmic}
\usepackage{acronym}
\newtheorem{lem}{Lemma}

\newcommand{\defeq}{\stackrel{\triangle}{=}}
\def\bOmega{\boldsymbol{\Omega}}

\newcommand{\e}{\mathsf{e}}
\newcommand{\jj}{\mathsf{j}}
\newcommand{\Herm}{\mathsf{H}}
\newcommand{\Trans}{\mathsf{T}}

\newcommand{\y}{\mathsf{y}}
\newcommand{\z}{\mathsf{z}}
\newcommand{\bx}{\mathbf{x}}

\newcommand{\bW}{\mathbf{W}}
\newcommand{\bs}{\mathbf{s}}

\newcommand{\bH}{\mathbf{H}}
\newcommand{\ba}{\mathbf{a}}
\newcommand{\bc}{\mathbf{c}}

\newcommand{\bd}{\mathbf{d}}

\newcommand{\br}{\mathbf{r}}
\newcommand{\bbm}{\mathbf{m}}
\newcommand{\bbM}{\mathbf{M}}

\newcommand{\bX}{\mathbf{X}}
\newcommand{\bh}{\mathbf{h}}

\newcommand{\bq}{\mathbf{q}}
\newcommand{\bv}{\mathbf{v}}
\newcommand{\bt}{\mathbf{t}}
\newcommand{\bz}{\mathbf{z}}

\newcommand{\bGamma}{\boldsymbol{\Gamma}}
\newcommand{\btau}{\boldsymbol{\tau}}
\newcommand{\blambda}{\boldsymbol{\lambda}}

\newcommand{\Ex}{\mathbb{E}}
\newcommand{\Var}{\mathbb{V}}

\newcommand{\tx}{\mathrm{tx}}
\newcommand{\thr}{\mathrm{thr}}
\newcommand{\SNR}{\mathrm{SNR}}

\def\bphi{\boldsymbol{\phi}}

\def\bomega{\boldsymbol{\omega}}
\def\bPsi{\boldsymbol{\Psi}}

\def\Cset{\mathbb{C}}
\def\Rset{\mathbb{R}}
\def\Zset{\mathbb{Z}}
\def\LOS{\mathrm{LOS}}
\def\nLOS{\mathrm{nLOS}}
\def\tmax{\mathrm{max}}
\def\tx{\mathrm{tx}}
\def\rx{\mathrm{rx}}
\def\BS{\mathrm{BS}}
\def\RIS{\mathrm{RIS}}

\def\SNR{\mathrm{SNR}}
\def\eff{\mathrm{eff}}

\def\thr{\mathrm{thr}}

\def\sCN{\mathcal{CN}}

\def\Wset{\mathcal{W}}
\def\bigO{\mathcal{O}}
\usepackage{xcolor}

%% file: Main.bbl
\begin{thebibliography}{99}

\bibitem{di2019smart}
M.~Di~Renzo \emph{et~al.}, ``Smart radio environments empowered by {AI}
  reconfigurable meta-surfaces: {An} idea whose time has come,'' vol. 129, May
  2019.

\bibitem{yu2021smart}
X.~Yu, V.~Jamali, D.~Xu, D.~W.~K. Ng, and R.~Schober, ``Smart and
  reconfigurable wireless communications: From RIS modeling to algorithm
  design,'' \emph{IEEE Wireless Communications}, vol.~28, no.~6, pp. 118--125,
  2021.

\bibitem{Tang2021}
W.~Tang \emph{et~al.}, ``Wireless communications with reconfigurable
  intelligent surface: Path loss modeling and experimental measurement,''
  \emph{IEEE Transactions on Wireless Communications}, vol.~20, no.~1, pp.
  421--439, 2021.

\bibitem{najafi2021}
M.~Najafi, V.~Jamali, R.~Schober, and H.~V. Poor, ``Physics-based modeling and
  scalable optimization of large intelligent reflecting surfaces,'' \emph{IEEE
  Transactions on Communications}, vol.~69, no.~4, pp. 2673--2691, 2021.

\bibitem{ghannam2021reconfigurable}
R.~Ghannam \emph{et~al.},
  ``Reconfigurable surfaces using fringing electric fields from nanostructured
  electrodes in nematic liquid crystals,'' \emph{Advanced Theory and
  Simulations}, vol.~4, no.~7, p. 2100058, 2021.

\bibitem{guirado2022mm}
R.~Guirado, G.~Perez-Palomino, M.~Ca{\~n}o-Garc{\'\i}a, M.~A. Geday, and
  E.~Carrasco, ``mm-wave metasurface unit cells achieving millisecond response
  through polymer network liquid crystals,'' \emph{IEEE Access}, vol.~10, pp.
  127\,928--127\,938, 2022.

\bibitem{aboagye2022design}
S.~Aboagye, A.~R. Ndjiongue, T.~M. Ngatched, and O.~A. Dobre, ``Design and
  optimization of liquid crystal RIS-based visible light communication
  receivers,'' \emph{IEEE Photonics Journal}, vol.~14, no.~6, pp. 1--7, 2022.
  
\bibitem{ndjiongue2021re}
A. R. Ndjiongue, T. M. N. Ngatched, O. A. Dobre, and H. Haas, ``Re-configurable intelligent surface-based VLC receivers using tunable liquid-crystals: The concept,'' \emph{Journal of Lightwave Technology}, vol. 39, no. 10, pp. 3193--3200, 2021, IEEE.


\bibitem{Robin2023Compact}
R.~Neuder, D.~Wang, M.~Schüßler, R.~Jakoby, and A.~Jiménez-Sáez, ``Compact
  liquid crystal-based defective ground structure phase shifter for
  reconfigurable intelligent surfaces,'' in \emph{2023 17th European Conference
  on Antennas and Propagation (EuCAP)}, 2023, pp. 1--5.

\bibitem{jimenez2022liquid}
A.~Jim{\'e}nez-S{\'a}ez, A.~Asadi, R.~Neuder, M.~Delbari, and V.~Jamali,
  ``Reconfigurable intelligent surfaces with liquid crystal technology: A
  hardware design and communication perspective,'' \emph{arXiv:2308.03065}, 2023, unpublished.

\bibitem{garbovskiy2012liquid}
Y.~Garbovskiy \emph{et~al.}, ``Liquid crystal phase
  shifters at millimeter wave frequencies,'' \emph{Journal of Applied Physics},
  vol. 111, no.~5, p. 054504, 2012.

\bibitem{neuder2023architecture}
R.~Neuder, M.~Sp{\"a}th, M.~Sch{\"u}{\ss}ler, and A.~Jim{\'e}nez-S{\'a}ez,
  ``Architecture for sub-100 ms liquid crystal reconfigurable intelligent
  surface based on defected delay lines,'' 2023, unpublished.

\bibitem{xu2014image}
D.~Xu \emph{et~al.}, ``Image sticking in liquid crystal displays with lateral electric
  fields,'' \emph{Journal of Applied Physics}, vol. 116, no.~19, 2014.

\bibitem{nobles2021eight}
J.~E. Nobles \emph{et~al.}, ``Eight-element liquid crystal
  based 32 GHz phased array antenna with improved time response,''
  \emph{Engineering Research Express}, vol.~3, no.~4, p. 045033, 2021.

\bibitem{ericksen1961conservation}
J.~L. Ericksen, ``Conservation laws for liquid crystals,'' \emph{Transactions
  of the Society of Rheology}, vol.~5, no.~1, pp. 23--34, 1961.

\bibitem{leslie1968some}
F.~M. Leslie, ``Some constitutive equations for liquid crystals,''
  \emph{Archive for Rational Mechanics and Analysis}, vol.~28, pp. 265--283,
  1968.

\bibitem{jakeman1972electro}
E.~Jakeman and E.~Raynes, ``Electro-optic response times in liquid crystals,''
  \emph{Physics Letters A}, vol.~39, no.~1, pp. 69--70, 1972.

\bibitem{wang2005studies}
H.~Wang, ``Studies of liquid crystal response time,'' \emph{University of Central Florida}, 2005.

\bibitem{wu1990experimental}
S.-T. Wu and C.-S. Wu, ``Experimental confirmation of the osipov-terentjev
  theory on the viscosity of nematic liquid crystals,'' \emph{Physical Review
  A}, vol.~42, no.~4, p. 2219, 1990.

\bibitem{hassanein2019optimizing}
G. Nabil Hassanein, ``Optimizing low frequency electro-optic response of nematic liquid crystals,'' \emph{Optik}, vol. 182, pp. 269--274, 2019, Elsevier.

\bibitem{tse2005fundamentals}
  D. Tse and P. Viswanath, ``Fundamentals of wireless communication,'' \emph{Cambridge university press}, 2005.

\bibitem{Zhang2020Capacity}
S.~Zhang and R.~Zhang, ``Capacity characterization for intelligent reflecting
  surface aided MIMO communication,'' \emph{IEEE Journal on Selected Areas in
  Communications}, vol.~38, no.~8, pp. 1823--1838, 2020.


\end{thebibliography}
